\documentclass[superscriptaddress,aps,pra,amssymb,amsfonts]{revtex4}

\usepackage{amsmath}
\usepackage{amsfonts}
\usepackage{hyperref}

\usepackage{graphicx,color,graphics}

\newcommand{\bra}[1]{\langle{#1} \vert}
\newcommand{\ket}[1]{\vert{#1} \rangle}

\def\one{{\mathchoice {\rm 1\mskip-4mu l} {\rm 1\mskip-4mu l} {\rm
1\mskip-4.5mu l} {\rm 1\mskip-5mu l}}}

\def\tr{{\rm tr}\; }
\def\cO{{\cal O}}

\def\({\left(}
\def\){\right)}

\newcommand{\braket}[2]{\langle#1|#2\rangle}

\newtheorem{lemma}{Lemma}
\newtheorem{theorem}{Theorem}
\newtheorem{definition}{Definition}
\newtheorem{proof}{Proof}

\begin{document}

\title{Spectral Gap Amplification}

\author{R. D.  Somma}
\affiliation{Los Alamos National Laboratory, Los
Alamos, NM 87545, USA}
\email {somma@lanl.gov}
\author{S. Boixo }
\affiliation{
Information Sciences Institute, USC, Marina del Rey, CA 90292, USA. Harvard University, Cambridge, MA 02138, USA.}
\email{sboixo@isi.edu}

\date{\today}

\begin{abstract}
 Many problems can be solved by preparing a 
specific eigenstate of some Hamiltonian $H$. The generic cost
of quantum algorithms for these problems is determined by the inverse
spectral gap of $H$ for that eigenstate and the cost of evolving with $H$
for some fixed time. The goal of spectral gap amplification is  to 
 construct a Hamiltonian $H'$ with the same eigenstate as $H$ but 
 a bigger spectral gap, requiring that constant-time evolutions with $H'$ and $H$
are implemented with nearly the same cost. 
 We show that a
quadratic spectral gap amplification is possible when $H$ satisfies
a frustration-free property and give $H'$ for these cases. 
This results in quantum speedups for optimization problems.
It also yields improved constructions for  adiabatic simulations
of quantum circuits and for the preparation of projected entangled pair states (PEPS),
which    play an important role in quantum many-body physics.
Defining a suitable
black-box model, 
we establish that the quadratic amplification is optimal 
for frustration-free Hamiltonians   and that no 
spectral gap amplification is possible, in general, if the frustration-free
property is removed. A corollary is that
finding a similarity transformation
between a stoquastic Hamiltonian and the corresponding stochastic matrix is hard in the black-box model,
setting  limits to the power of some classical methods that simulate quantum adiabatic evolutions.
\end{abstract}

\maketitle 



\pagestyle{myheadings}
\thispagestyle{plain}
\markboth{R. D. Somma AND S. Boixo }{Spectral Gap Amplification}

\section{Introduction and summary of results}
Many problems in physics and optimization,
such as describing quantum phases of matter or solving satisfiability,
can be reduced to the computation of low-energy states of Hamiltonians~(see, e.g.\cite{sachdev_2001,farhi:qc2001a}). Methods to compute such states exist,
with adiabatic state transformation (AST) being one of the 
most acknowledged and powerful heuristics for that goal~\cite{amara:qc1993a,
finnila:qc1994a,kadowaki:qc1998a,farhi_quantum_2000,farhi:qc2001a,farhi_quantum_2002}.
The problem in AST involves transforming a pure quantum state 
$\ket{\psi_0}$ into $\ket{\psi_1}$; these  being the endpoints of a
state path $\ket{\psi_s}$, $0 \le s \le 1$, where each $\ket{\psi_s}$ is an eigenstate of a Hamiltonian
$H(s)$.
Quantum strategies to solve the AST problem typically require evolving with the Hamiltonians for time $T>0$,
which determines the cost of the method. A well-known example
is given by the quantum adiabatic theorem and requires  changing the interaction parameter $s$ 
sufficiently slowly in time~\cite{messiah_1999,jansen_bounds_2007}, but more efficient
strategies exist~\cite{boixo:qc2009a,wocjan_speed-up_2008,boixo:qc2010a}.
When direct access to Hamiltonian evolutions is not possible, such evolutions can be simulated in the standard quantum circuit model
by means of Trotter-Suzuki-like formulas~\cite{berry_efficient_2007,wiebe_product_2010,papageorgiou_efficiency_2010}.
For these cases, the cost of the method does not only depend on $T$, but also depends on
the additional cost of approximating fixed-time Hamiltonian evolutions by  quantum circuits. 
For this reason, it is standard 
to restrict to those $H(s)$ whose  fixed-time evolutions
can be approximated by quantum circuits of small size (e.g., size bounded by
constant-degree polynomials).

Let $L$ be the angular length determined by $\{ \ket{\psi_s}
\}_{0 \le s \le 1}$ and $\Delta_s \ge \Delta$ the spectral gap to
the other nearest eigenvalue of $H(s)$.  Recently, we provided fast
quantum methods for eigenpath traversal that do not exploit the structure
of the Hamiltonians and solve the AST problem in  time
$T \in \Omega(L/\Delta)$~\cite{boixo:qc2010a}. Such methods 
were shown to be optimal in a black-box model~\cite{boixo:qc2009b}.
However, the results in  Ref.~\cite{boixo:qc2009b} do not forbid
the existence of faster methods for solving the AST
problem if additional knowledge about the structure of $H(s)$ is available. Our paper
focuses on this observation and introduces the spectral gap  amplification problem or GAP.
Basically, we intend to answer:
Can we construct  a Hamiltonian $H'(s)$ that also has $\ket{\psi_s}$ as
eigenstate, but spectral gap $\Delta'_s \gg \Delta_s$?

When spectral gap amplification is possible, it can be used as a
technique to find new quantum
speedups in both, the Hamiltonian-based model and the quantum circuit model. In Refs.~\cite{somma_quantum_2008,somma_optimization_2010}, for
example, we constructed the so-called quantum simulated annealing
algorithm that provided speedups of the well-known simulated annealing
method implemented using Monte-Carlo. The reason for such speedups is
a gap amplification step: if $\Delta$ is the spectral gap of the
stochastic matrix $S$ used in Monte-Carlo, there is a Hamiltonian $H'$
with spectral gap $\Delta' \ge \sqrt{\Delta}$ and eigenstate that
allows to sample from the fixed point of $S$.  To build $H'$ we used
previously-developed tools for quantum
walks~\cite{szegedy_quantum_2004,ambainis_walk_2004}.

Motivated by the results
in Refs.~\cite{somma_quantum_2008,somma_optimization_2010} we study 
the GAP in different scenarios. We first show that, for
Hamiltonians that satisfy a frustration-free
property~\cite{perez_PEPS_2008,bravyi_stoquastic_2009,beaudrap_frustfree_2010,wolf_phase_2006},
a quadratic spectral gap amplification is possible and {\em
 optimal} in some suitable black-box model. Quadratic spectral gap amplification was previously known
only for Hamiltonians that result from a similarity transformation of
stochastic
matrices~\cite{szegedy_quantum_2004,somma_quantum_2008,somma_optimization_2010}. These
Hamiltonians, which are the so-called discriminants, are also stoquastic, i.e. the off-diagonal entries are
non-positive.  A direct consequence of our new construction is that spectral
gap amplification can now be achieved for any frustration-free Hamiltonian as those used in quantum adiabatic
simulations of quantum circuits, improving upon the results in Refs.~\cite{aharonov_adiabatic_2007,mizel_equivalence_2007,aharonov_line_2009}.
The eigenstate in our construction
is not the lowest-energy state of the Hamiltonian induced by the quantum circuit (which was the case
in Refs.~\cite{aharonov_adiabatic_2007,mizel_equivalence_2007,aharonov_line_2009}).
Nevertheless,  techniques for AST work for any eigenstate and still apply to our case, giving more efficient
quantum adiabatic simulations to prepare the circuit's output state. 
In addition,
some low-energy eigenstates of general frustration-free Hamiltonians (termed PEPS) were  shown
to play an important role in physics, renormalization, and optimization~\cite{verstraete_peps_2006,wolf_phase_2006,perez_PEPS_2008}:
PEPS are useful variational states
reproducing the local physics with high accuracy~\cite{verstraete_peps_2004,verstraete_peps_2006}. 
Our method for spectral gap amplification can be used to speedup the preparation
of PEPS on a quantum computer, with an expected improvement of other methods
for this goal (e.g.,~\cite{schwarz_peps_2011}).

Interestingly,
it is hard to achieve spectral gap amplification 
for general Hamiltonians that do not satisfy the frustration-free property.
It would otherwise imply fast quantum methods
for  some problems,
 contradicting known complexity bounds in
the oracle setting~\cite{bennett_searchbound_1997}.
In addition, 
a natural generalization of our technique for spectral gap amplification 
would output a Hermitian operator only if the input Hamiltonian $H$
is semidefinite positive (as is the case for frustration-free Hamiltonians).
Such a generalization could then be used to solve the decision problem of 
whether a given $H$ is semidefinite positive or not
(with some small error bound), and
provide an estimate to
the lowest eigenvalue of $H$.
This is a complete problem for the class QMA~\cite{kitaev_computation_2002},
and finding its solution may require a quantum circuit of undesirably large size.

The result on the impossibility of general spectral gap amplification 
 has additional implications
on the power of classical methods that solve the AST problem.
For example,
some stoquastic Hamiltonians
can be mapped to stochastic matrices
through a similarity transformation~\cite{somma_thermod_2007,bravyi_stoquastic_2008,bravyi_stoquastic_2009}. 
The fixed point of the stochastic matrix coincides
with the distribution induced by the ground state
of the Hamiltonian.  The stochastic matrix would allow us to build a Monte-Carlo
method to sample
from the fixed point and classically solve the AST problem in this case.
The same similarity transformation could then be used
to give a frustration-free representation of the Hamiltonian~\cite{verstraete_peps_2006,somma_thermod_2007},
allowing us to use the results 
for frustration-free Hamiltonians to amplify the gap.
This contradicts the impossibility result, 
implying that finding the transformation is hard 
(in the black-box setting).

We organize the manuscript  as follows. First, in Sec.~\ref{conventions}
we define the GAP in more detail and
present some methods for the proofs of the quadratic amplification in the frustration-free case,
and for the impossibility of amplification in general. In Sec.~\ref{frustrationfree} we show the construction
of $H'$ for frustration-free Hamiltonians. In Sec.~\ref{blackbox} we prove that the cost of implementing evolutions 
under $H'$ is nearly the same as the cost of implementing evolutions under $H$ in the black-box
model, and present 
a simulation method. In Sec.~\ref{optimal} we prove that the quadratic amplification is optimal 
for frustration-free Hamiltonians in this black-box model. In Sec.~\ref{adiabaticsimulation}
we comment on the advantages of our gap amplification construction for the adiabatic simulation
of quantum circuits and give a local Hamiltonian. In Sec.~\ref{generalcase} we prove that no gap amplification
is possible for general Hamiltonians in the black-box model. This results in limits
on the power of classical simulations for quantum systems, that are discussed in
Sec.~\ref{stochastic-stoquastic}.
We conclude in Sec.~\ref{conc}.

\section{Definitions and methods}
\label{conventions}

We first comment on the  implementation cost of Hamiltonian evolutions. 
Ignoring precision, we assume Hamiltonians $H$ given as black boxes
that implement the evolution $ \exp \{ - i H t \}$ on ancilliary input $\ket t$,
satisfying $|t| \le \pi$. When we write $H = \sum_k \Lambda_k$ for the Hamiltonian,
we assume a modified black-box that implements  $\exp \{ - i \Lambda_k t \}$
on ancilliary input $\ket k \ket t$. Such a black box is typically an efficient quantum circuit in applications.
 We write $O_H$ for the black-box.
The implementation cost of $\exp \{ - i H' s \}$ in the black-box model,
where $H'$ is the Hamiltonian resulting from the gap amplification construction
applied to $H$, will be the 
number of times $O_H$ is called to implement the operation in either case.

\begin{definition}
 Let $ H = \sum_k \Lambda_k  \in \mathbb{C}^N \times \mathbb{C}^N$ be a finite-dimensional Hamiltonian, $\ket \psi$ a (unique)
  eigenstate of $H$ with eigenvalue $\lambda$, and $\Delta$ the  
  gap to the other nearest eigenvalue of $\lambda$; i.e. $\Delta$ is the spectral gap.  The
  goal of the GAP is to find a new Hamiltonian $H'$ that has $\ket \psi
  \otimes \ket {\mathfrak{0}}$ as (unique) eigenstate and spectral
  gap $\Delta' \in \Omega(\Delta^{1-\epsilon})$, for $ \epsilon >
  0$.  $\ket{\mathfrak{0}}$ is any subsystem's state for which an
  efficient quantum circuit is known.  The implementation cost of $\exp\{-i H' t \}$ in the black-box model
  must be of the same order as the implementation cost of $\exp\{- i H t
  \}$.
\end{definition}

The last requirement forbids naive constructions such as $H' = c H$
for $c \gg 1$ and is needed to carry quantum speedups from the
Hamiltonian-based model, in which the cost is the evolution time, to
the black-box model, in which the cost is the number of black boxes.
If both $H$ and $H'$ are sparse with a bounded number of (efficiently computable)
non-zero entries per row and matrix oracles
for the $\Lambda_k$ are provided, the last requirement can be
satisfied~\cite{somma_physics_2002,aharonov_adiabatic_2003,
  berry_efficient_2007, childs_efficient_2010}.

\begin{definition}
  A Hamiltonian $H \in \mathbb{C}^N \times \mathbb{C}^N$ is
  frustration free if it can be written as  $H = \sum_{k=1}^L
  a_k \Pi_k$, with (known) $0 \le a_k \le 1$, $(\Pi_k)^2 =
  \Pi_k$ projectors, and $L \in \cO[\text{polylog}(N)]$.  Further, if
  $\ket \psi$ is a lowest-eigenvalue eigenstate (ground state ) of
  $H$, then $\Pi_k \ket \psi = 0 \ \forall \ k$.
\end{definition}

$\ket \psi$ is then a ground state of every term in the decomposition
of $H$. 
We will assume that $a_k = 1 \ \forall
k$ when there is no loss of generality. Then $H = \sum_k \Pi_k$ is still frustration free, it has $\ket
\psi$ as ground state, it has at least the spectral gap of $\sum_k a_k \Pi_k$, and $\| H \| \le L$. While the $\Pi_k$'s may be
local operators, we do not make that assumption here. It will suffice
to have access to the black box $O_H$.

We consider the unstructured search problem or SEARCH  to prove
the optimality of quadratic gap amplification for frustration-free Hamiltonians and 
also the impossibility
of gap amplification for general Hamiltonians. We introduce SEARCH and describe  a simple 
variant of Grover's algorithm to solve it.
\begin{definition}\label{def:search}
Consider a family of Boolean functions or \emph{oracles} $\{f_x\}$, with
domain $[0,..,M-1]$, that satisfy $f_x(y)=\delta_{xy}$. 
The goal of SEARCH is
to find the unknown word $x$, with high probability,
by querying $f_x$ the least possible number of times.
\end{definition}

For quantum algorithms,
we use a reversible version of $f_x$ defined as
\begin{align}
\label{searchoracle}
Q_x \ket{y} \ket b=  \ket{y} \ket {b \oplus f_x(y)} \; ,
\end{align}
with $b \in \{0,1\}$.
We refer to $Q_x$ as the search  oracle.
Quantum algorithms that solve SEARCH
with $\Theta(\sqrt{M})$ uses of
$Q_x$, and other $x$-independent operations, exist~\cite{bennett_searchbound_1997,
farhi_analog_1998}. Classical algorithms require $\Theta(M)$
queries to $f_x$.

From Ref.~\cite{childs_quantum_2002}
we can give a quantum algorithm to solve SEARCH
using projective measurements.
\begin{lemma}
\label{lemma:random_grover}
  Let $\{H_x\}$ be a family of Hamiltonians  
  with unique eigenstates $\ket{\psi_x}$ satisfying $p_x=|\bra x \psi_x \rangle|^2 \in \cO(1)$ and $p_s=|\bra s \psi_x \rangle |^2 \in \cO(1)$, where $\ket s = \frac 1 {\sqrt{M}} \sum_{y=0}^{M-1} \ket y$ is the equal superposition state.
 A quantum algorithm can find $x$
  with cost $\cO(1/\Delta)$, where $\Delta$ is a bound on the
  eigenvalue gap of $H_x$.
  \end{lemma}
  \begin{proof}
    The probabilities $p_x$ and $p_s$ are bounded independently of $M$.
     If we initialize a quantum computer in $\ket s$ and next
    we perform two projective measurements of $\ket{\psi_x}$ and $\ket
    x$, respectively, we can find $x$ with probability no smaller than
    $p_x. p_s $.  The cost of this method is measured by $T$, the time of evolving
    with $H_x$. This is needed to simulate the measurement of
    $\ket{\psi_x}$, with high accuracy, using the well-known phase estimation
    algorithm~\cite{kitaev_quantum_1995,childs_quantum_2002} or by phase
    randomization~\cite{boixo:qc2009a}. The  cost $T$ is dominated by the
    inverse gap of $H_x$. In optimal constructions, $\Delta \in \cO(1/\sqrt{M})$ so that
    $T \in \cO(\sqrt{M})$.
  \end{proof}

\section{GAP:   Frustration-free Hamiltonians}
\label{frustrationfree}
When a Hamiltonian $H=\sum_k \Pi_k$ is frustration free, local, and
stoquastic~\cite{bravyi_stoquastic_2008}, we could use the results
in Ref.~\cite{bravyi_stoquastic_2009}
and Refs.~\cite{somma_quantum_2008,somma_optimization_2010} to build a
Hamiltonian $H'$ with a quadratically bigger spectral gap. In that case, $H$ is
the so-called discriminant of a Markov process or stochastic
matrix.  Such a construction uses tools from quantum walks~\cite{szegedy_quantum_2004}
 to speed up Markov chain based
algorithms and requires full knowledge of each $\Pi_k$. 
 Here, we  present a novel and more efficient technique that allows
for quadratic gap amplification of any frustration-free Hamiltonian, without
the requirement of $H$ being local and stoquastic. Later we show that this technique provides
the optimal amplification
and present a black-box simulation method for the evolution 
 with $H'$ that does not require full knowledge of each $\Pi_k$.

\begin{theorem}\label{th:gap_amp_ff}
  Let $H=\sum_{k=1}^L \Pi_k \in \mathbb{C}^N \times \mathbb{C}^N$ 
  be a frustration-free Hamiltonian, $\ket \psi$ a ground state,
  and $\Delta$ the spectral gap or smallest nonzero eigenvalue.
 Then, there exists
  $H'$ satisfying $H' \ket \psi \ket{\mathfrak{0}}=0$ and eigenvalue gap $\Delta' \in \Omega(\sqrt{\Delta  L})$. The dimension of the null space of
$H'$ is that of the null space of $H$.
\end{theorem}

\begin{proof}
The proof is constructive. 
We let
\begin{align}
\label{eq:Xdef}
 X = \sum_{k=1}^L \Pi_k \otimes \Upsilon_k 
\end{align}
 be a Hamiltonian, and $\Upsilon_k = \ket k \bra k$
a rank-one projector acting on $\mathbb{C}^L$. It is easy to show that
$X^n=X$ for $n\ge 1$. We  define the unitary  operator
\begin{align}
U = e^{- i \pi X} = \one - 2 \sum_{k=1}^L \Pi_k \otimes \Upsilon_k \; ,
\end{align}
with $\one$ the identity matrix whose dimension will be clear from context
(here it is $NL \times NL$).
The implementation of $U$ requires
a single black box $O_H$ and $U^2 = \one$.
In addition,  
$\ket{\mathfrak {0}} =\frac 1 {\sqrt{L}} \sum_{k=1}^L  \ket k$
is an ancilliary state and  $P=\one \otimes \ket{\mathfrak {0}}\bra{\mathfrak {0}}$ its projector.
Because $L \in \cO[\text{polylog} (N)]$,  quantum circuits that prepare $\ket{\mathfrak {0}}$
efficiently exist~(e.g. \cite{somma_physics_2002}).

We define the Hamiltonian
\begin{align}
G &= UPU-P \; .
\end{align}
If  $\{ \ket{ \psi_j} \}_{1 \le j \le N}$ are
the eigenstates of $H$   with eigenvalues  $\lambda_1 =0 \le
\lambda_2 \le \ldots \le \lambda_N \le L$,  $G$ leaves invariant
the at most two-dimensional subspace ${\cal V}_j$ spanned by
$\{ \ket{\psi_j} \ket{\mathfrak{0}} , U \ket{\psi_j} \ket{\mathfrak{0}} \}$,
 for all $j$. To show this, we 
need an important property that is a consequence
of our construction:
\begin{align}
\label{property1}
PUP= \left( \one -\frac 2 {L} H \right ) \otimes \ket{\mathfrak {0}}\bra{\mathfrak {0}} = A \otimes \ket{\mathfrak {0}}\bra{\mathfrak {0}} \; .
\end{align}
The eigenvalues of  $A$ are $\gamma_j =1  -2 \lambda_j /L$
and the eigenstates are also the $\ket{\psi_j}$.
To avoid technical issues we assume that
$\gamma_N \ge 0$. If this were not the case, the projector $P$ could be easily
modified so that $\gamma_j =1  - \lambda_j /L$ instead, and the assumption is fulfilled. 
We can think of $\gamma_j$ 
as being $\cos \alpha_j= \bra{\psi_j} \bra{\mathfrak{0}} PUP \ket{\psi_j} \ket{\mathfrak{0}} $, so that $\alpha_j$ is the
angular distance between the states $\ket{\psi_j} \ket{\mathfrak{0}}$ and $U\ket{\psi_j} \ket{\mathfrak{0}}$
(Fig.~\ref{fig:vectorrep}).
When $j$
is such that $\lambda_j=0$, we have
$U \ket{\psi_j} \ket{\mathfrak{0}} =\ket{\psi_j} \ket{\mathfrak{0}}$ and ${\cal V}_j$
is one dimensional. This implies that  $ \ket{\psi_j} \ket{\mathfrak{0}}$ is an eigenstate of $G$
of  eigenvalue 0. When
$\lambda_j \ne 0$, Eq.~(\ref{property1}) gives $G \ket{\psi_j} \ket{\mathfrak{0}} = \gamma_j U \ket{\psi_j} \ket{\mathfrak{0}}
- \ket{\psi_j} \ket{\mathfrak{0}}$, and $G$ is invariant in ${\cal V}_j$.

\begin{figure}[ht]
  \centering
  \includegraphics[width=10cm]{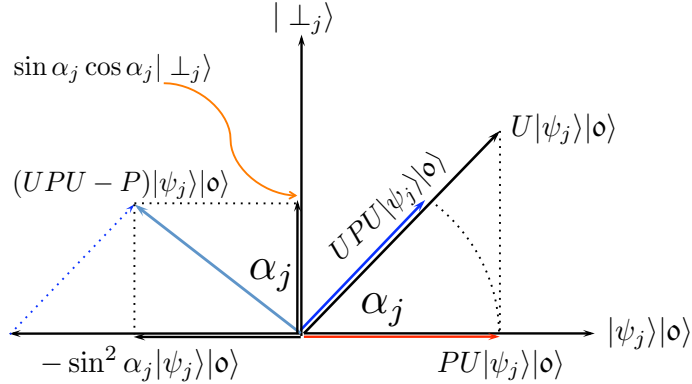}
 \caption{
Geometric representation of the action of $G$ on $\ket{\psi_j} \ket{\mathfrak{0}}$.}
  \label{fig:vectorrep}
\end{figure}

When $\lambda_j \ne 0$ (or $\alpha_j \ne 0$), 
we define $\ket{\perp_j}$ as the unit state orthogonal to $\ket{\psi_j} \ket{\mathfrak{0}}$
in the subspace ${\cal V}_j$. 
We can use simple
geometry arguments (Fig.~\ref{fig:vectorrep}) to show that the two-dimensional block of $G$ in the basis 
$\{ \ket{\psi_j} \ket{\mathfrak{0}} ,\ket{\perp_j}\} $ is
\begin{align}
G_j= \begin{pmatrix} -  \sin^2 \alpha_j & \sin \alpha_j \cos \alpha_j \cr
\sin \alpha_j \cos \alpha_j & \sin^2 \alpha_j
\end{pmatrix} \; .
\end{align}
The eigenvalues of $G$ in this subspace are
 $\pm \sin \alpha_j$. Because  $\sin \alpha_j \ge \sqrt{2 \lambda_j/L}$,
 a quadratic amplification in the spectral gap follows.

 We need to account for a small technical issue
 to obtain  $H'$. $G$ acts trivially on any other eigenstate $\ket \phi$ orthogonal to
 $\oplus_j {\cal V}_j$. The dimension of the null space of $G$ is much larger than
 the dimension of the null space of $H$. Usually, this poses no difficulty in the AST
 problem because $G$ does not generate transitions between the states $\ket{\psi_j} \ket{\mathfrak{0}}$ and any other state in the null space of $G$.  
 In those cases we can assume $H' =L G$ so that $\|H' \| \le L$. Otherwise, 
 since    $ P \ket \phi 
 =( \sum_j \ket{\psi_j}
 \bra{\psi_j} \otimes \ket{\mathfrak{0}} \bra{\mathfrak{0}} ) \ket \phi  = 0$,
 we can easily remove the unwanted degeneracy of $G$ by adding a ``penalty term'' to those states
 whose ancilliary register is not in $\ket{\mathfrak{0}}$. That is,
 \begin{align}
 H'=L G + \delta (\one -P)
 \end{align}
 for some $\delta >0$ determined below.
 Note that, if $\lambda_j=0$,
  $\ket{\psi_j} \ket{\mathfrak{0}}$ is still in the null space of $H'$.
 For $\lambda_j \ne 0$, the matrix  of $H'$ in ${\cal V}_j$ is
 \begin{align}
H'_j= L \begin{pmatrix} -  \sin^2 \alpha_j & \sin \alpha_j \cos \alpha_j \cr
\sin \alpha_j \cos \alpha_j & \delta/L + \sin^2 \alpha_j 
\end{pmatrix} \; .
\end{align}
It suffices to choose $\delta =\sqrt{\Delta  L}$ to 
guarantee that $\xi_j ^\pm$, the eigenvalues of $H'$ in this subspace,
satisfy $|\xi_j^\pm| \ge(1/6) \sqrt{\lambda_j L}$.
The gap of $H'$ is then $\Delta ' \in \cO(\sqrt{\Delta L})$
and the dimension of the null space of $H'$ coincides with that of the null space of $H$.
\end{proof}

Our choice of $\delta$ uses $\Delta$, which may be unknown. 
We can then safely use a lower bound on the spectral gap of $H$,
$\bar \Delta \le \Delta$,  
and build $H'$ using $\bar \Delta$ instead.
The loss is that the resulting gap of $H'$ may not be
the largest possible. However, this is not a drawback if we use
$H'$ (or $LG$) for the AST problem. As mentioned, the relevant
spectral gap is of order $\sqrt{\Delta L}$ even if $\Delta$
is unknown, and the role of $\bar \Delta$ is unimportant in this case.

\subsection{Black-box simulation}
\label{blackbox}
In a Hamiltonian-based model of quantum computation, 
a gain on the spectral gap readily provides a quantum speedup for the AST problem. This 
model assumes the ability of evolving under $H'$, i.e. applying $V(t)=\exp\{-i H' t\}$ directly. 
To carry those speedups to the quantum circuit model, we need a way of approximating
a fixed-time evolution by a short sequence of gates. Usually, such a sequence
is obtained by using Trotter-Suzuki-like approximations~\cite{berry_efficient_2007,wiebe_product_2010}.
In those approximations, the evolution of a Hamiltonian $H=\sum_k \Lambda_k$ is built upon
short-time evolutions under each $\Lambda_k$, and each such evolution is written
as a constant-size sequence of gates~\cite{somma_physics_2002}. In other words,
what determines the size of the circuit that approximates $V(t)$ is the implementation cost 
as given by the number of uses of $O_H$, the black box described in Sec.~\ref{conventions}.
Here we describe a simulation method for $V(t)$ whose overall cost is almost 
linear in $L |t|$. In most cases this result  implies a quantum speed-up, in the circuit model, 
 for those methods that solve the AST problem using paths of frustration-free
 Hamiltonians.

\begin{lemma}\label{lem:hp_evolution}
  For any $\kappa$ and $\epsilon$, there is
  a black-box simulation $W(t)$ of $V(t)$ that requires
  using the black box  $\cO\(  |L t|^{1 + 1/(2 \kappa)}\)$ times, $\kappa \in \mathbb{N}^*$,
  and satisfies
 $ \| W(t) -V(t) \| \le \epsilon$.
\end{lemma}
 \begin{proof}
   Up to a trivial offset, 
   $H' = A_1 - (L+\delta)P$, with $A_1 =L UPU$, $\| A_1 \|=L$, and
   $\|P \| = 1$. We can then use known results
   on Hamiltonian simulations~\cite{aharonov_adiabatic_2003,berry_efficient_2007,childs_efficient_2010,papageorgiou_efficiency_2010}
   to approximate $\exp \{ - i H' t \}$ at precision $\epsilon$ by a sequence of
   $c(\epsilon,\kappa) |Lt|^{1 + 1/\kappa }$ concatenated evolutions under $A_1$ and $P$.
  For fixed $\epsilon$ and $\kappa$, $ c(\epsilon,\kappa)$ is constant.
   The proof follows by noticing that
   \begin{align}
   e^{-i A_1 s} = U e^{-i L P s } U \; ,
   \end{align}
   i.e. it requires two uses of $U$, and each $U$ can be implemented using $O_H$ once on input $\ket \pi$.
   Evolutions under $P$ for any time do not require the black box and can be implemented
  with a quantum circuit of size $L$ or smaller.
    \end{proof}

\subsection{Optimal amplification} 
\label{optimal}
We are interested in finding the biggest spectral gap amplification
possible for frustration-free $H$. To this end, we consider the set $\cal F$ of
all those $H'$ that have $\ket \psi \ket{\mathfrak{0}}$ as eigenstate and
 such that 
evolving with $H'$ for time $t$ can  be done
using $\cO(| t|)$ black boxes $O_H$; i.e. $L$ is constant. 
Under these assumptions, the quadratic spectral gap amplification is optimal in the black-box model:
\begin{theorem} 
  $\Delta' \in
  \Theta(\sqrt{\Delta})$. 
\end{theorem}

\begin{proof}
We sketch the proof here and leave the full version in
Appendix~\ref{App:OptimalityProof}. We consider instances $\{\tilde H_x \}$,
$x \in \{0,1,\ldots M-1\}$ unknown, that can be used to solve SEARCH as described 
in    Lemma~\ref{lemma:random_grover}. Each $\tilde H_x$
is frustration free and the spectral gap is of order $1/M$. The corresponding black box 
$O_{\tilde H_x}$ can be implemented with a single call to 
the search oracle of Eq.~(\ref{searchoracle}). If $\tilde H'_x$ is the Hamiltonian
with the amplified gap $\Delta'$, SEARCH can be solved by evolving
with $\tilde H'_x$ for time $\Omega(1/\Delta')$. By assumption, this would require using
$\Omega(1/\Delta')$ black boxes $O_{H_x}$  and thus $\Omega(1/\Delta')$ search oracles.
The lower bound on SEARCH implies
 $\Delta ' \in   \Theta(\sqrt{\Delta})$, with $\sqrt{\Delta} \in \cO( 1/\sqrt{M})$.
  \end{proof}

\section{Gap amplification and adiabatic quantum computation}
\label{adiabaticsimulation}
A  property of our  spectral gap amplification method is that it yields  quantum speed-ups
for the AST problem. It is also important to describe
additional advantages of the method, including the simulation of quantum circuits. 
It is known that any quantum circuit specified by quantum gates $U_1 , \ldots , U_Q$
can be simulated, in a Hamiltonian-based model, by an adiabatic quantum evolution involving
frustration-free Hamiltonians $H(s) = \sum_{k=1}^L a_k (s) \Pi_k(s)$. For $s=1$,
the ground state of $H(1)$ has large overlap with the state output by the quantum
circuit. In some constructions, $L \in \cO[\text{poly}(Q)]$ and $\Pi_k$ denotes nearest-neighbor (two-body)
interactions between spins of corresponding many-body systems in one or two-dimensional 
lattices~\cite{aharonov_line_2009,aharonov_adiabatic_2007}. 

The cost of the adiabatic
simulation  depends on the inverse minimum gap of $H(s)$, $\Delta$.
We can therefore use our gap amplification construction for frustration-free Hamiltonians to 
speed-up the  adiabatic simulation. However, in doing so,
we encounter difficulties regarding the locality of $H'$. For example, the projector $P$
can be interpreted as an interaction term involving $\log_2 (L)$ qubits and this is undesired
if we want to find a physical $H'$ whose interactions address a limited number of subsystems.
A simple procedure to avoid many-body terms is to consider each state $\ket {k}$ in $\ket{\mathfrak{0}}$
as a state of $L$ qubits of the form
\begin{align}
\label{eq:singleparticlemap}
\ket k \rightarrow \ket {0 \ldots 0 \! \! \! \! \! \! \underbrace{1}_{\text{ $k$th position}}\! \! \! \!  \! \! 0 \ldots 0 } \; .
\end{align}
The subspace defined by Eq.~(\ref{eq:singleparticlemap}) is typically termed the {\em
single-particle} subspace.
The terms $\ket k \bra{k'} $ appearing in $P$ map to $\sigma^+_k \sigma^-_{k'}$,
with $\sigma^{\pm} =( \sigma^x \pm i \sigma^y)/2$, and $\sigma^{\alpha}$ the corresponding Pauli operator.
The terms $\ket k \bra k $ map to $(\one - \sigma^z_k)/2$. In this subspace, 
\begin{align}
G =\frac 1 L\left [  \sum_{k,k'} e^{i \pi \Pi_k (\one - \sigma^z_k)/2} \sigma^+_k \sigma^-_{k'} e^{-i \pi \Pi_{k'} (\one - \sigma^z_{k'})/2}
- \sigma^+_k \sigma^-_{k'} \right] \; .
\end{align}
If the $\Pi_k$ are two-subsystem interaction terms, $H'=LG + \delta(\one -P)$ will map to 
a combination of four-subsystem terms. This makes $H'$ local
but not spatially local. We also note that, while the overall spectrum of $H'$ may have changed because of the mapping,
the spectrum of $H'$ in the single-particle subspace corresponds to  the one analyzed
in Sec.~\ref{frustrationfree}.

We note that our construction to amplify the gap is not
unique. Indeed, as seen in Appendix~\ref{App:LocalityProof}, the Hamiltonian
\begin{align}
\tilde G =  \sum_{k=1}^L  \Pi_k \otimes ( \ket k \bra{{0}} + \ket 0  \bra k)  
\end{align}
has a similar property: its eigenvalues are $\pm \sqrt{\lambda_j}$. Note that if $a_k \ne 1$ we can modify it as 
$\tilde G=  \sum_{k=1}^L \sqrt{a_k} \; \Pi_k \otimes ( \ket k \bra{{0}} + \ket 0  \bra k)$.
The action of  $\tilde G$ in the single-particle subspace is equivalent to that of
\footnote{The corresponding states are still those in  Eq.~(\ref{eq:singleparticlemap}), 
but each having $L+1$ qubits labeled from $0$ (left) to $L$ (right).}
\begin{align}
\tilde G = \sum_{k=1}^L \Pi_k \otimes ( \sigma^+_k \sigma^-_0 +  \sigma^-_k \sigma^+_0)  \;.
\end{align}
To eliminate the degeneracy of the null eigenvalue in this subspace, we can proceed as in Sec.~\ref{frustrationfree} and add a term proportional to $(\one - \ket 0 \bra 0)\rightarrow (1+\sigma^z_0)/2$, so that
\begin{align}
\tilde G  \rightarrow  \bar G= \sum_{k=1}^L \Pi_k \otimes ( \sigma^+_k \sigma^-_0 +  \sigma^-_k \sigma^+_0)
 +\frac{ \sqrt{\Delta}}{2} \frac{ (1+\sigma^z_0)}{2}  \; .
\end{align}
The relevant spectral gap of $\bar G$ in the single-particle subspace is still of order $\sqrt{\Delta}$.

Nevertheless, an unwanted degeneracy of the null eigenvalue from other (many-particle) 
states is  possible. (Note that any eigenstate in the 0-particle subspace
has eigenvalue exactly $\sqrt{\Delta}/2$).
 To remove such a degeneracy we set
 ``penalties'' to states that belong to different particle subspaces. To this end,
we  note that
\begin{align}
\nonumber
0&  \le \tilde G^2 = \sum_{k,k'} \Pi_k \Pi_{k'} (\sigma^+_k \sigma^-_{k'} \sigma^-_0 \sigma^+_0 + H.c.) \\ 
& \le S^+ S^- \sigma^-_0 \sigma^+_0 + H.c. \; ,
\end{align}
where $S^+ = \sum_{k=1}^L \sigma^+_k$ is proportional to the so-called $SU(2)$ spin {\em raising} operator.
Let $\ket{S,m}$ be the eigenstates of $S^z=\sum_{k=1}^L \sigma^z_k$, with eigenvalue $m$, and total $SU(2)$ spin proportional to $S$. Then, $ m \in \{ -L , -L+2, \ldots, L\}$, $0 \le S \le L$, and
\begin{align} 
\nonumber
S^+ \ket{S,m-2} & = N(S,m)\ket{S,m} \; , \\ 
\label{eq:SU(2)}
 N(S,m) & = \frac 1 2 \sqrt{(S+m) (S-m+2)} \; .
\end{align}
It implies
\begin{align}
\label{eq:SU(2)eigenvalue}
 \(S^+ S^- \sigma^-_0 \sigma^+_0  + H.c. \) & \ket{b}_0 \otimes \ket{S,m} =  \\
 \nonumber
 & = \( \delta_{b,0}
N(S,m+2)^2 + \delta_{b,1} N(S,m)^2 \)\ket{b}_0 \otimes \ket{S,m} \; ,
\end{align}
where $b \in \{0,1\}$ denotes the state of the  qubit in the $0$-th position. The operator
$\tilde G^2$ leaves the  $a$-particle subspace invariant, 
and we write $a=(L-m)/2$.
From Eqs.~\eqref{eq:SU(2)} and~\eqref{eq:SU(2)eigenvalue}, 
the eigenvalues of $\tilde G^2$ in that subspace are bounded from above by
$(L+m+2)(L-m+2) /4 = (L+1-a)(1+a)$. The upper bound to the absolute
value of the eigenvalues of
$\tilde G$, in that subspace, is $\sqrt{(L+1-a)(1+a)}$.

We define $Z=L-2-\sum_{k=0}^L \sigma^z_k$ and note that $Z$ acts trivially in the single-particle
space and nontrivially elsewhere. We also define
\begin{align}
H' = \frac{ 1 }{L^{}} \; \bar G + 2 Z \; ,
\end{align}
$d>0$. $H'$ acts invariantly in the $a$-particle subspaces. The eigenvalues of $H'$
are then bounded from below by
\begin{align}
-\frac{ 1 }{L^{1/d}}  \sqrt{(L+1-a)(1+a)} + 4(a-1) \; .
\end{align}
For the specific case of $d=2$, $a=1$, we already know that the eigenvalues of $H'$
are in the range $[-1, 1]$; this case  corresponds to the single-particle one
whose spectrum was analyzed previously. Also, for $d \le 2$, $a \ge 2$,
all the eigenvalues coming from the $a$-particle subspaces do not {\em mix}
with the eigenvalues of the single-particle subspace, and the potentially high degeneracy 
of the null space is avoided. For $a=0$ the exact eigenvalues of $H'$ are
$\sqrt{\Delta}/(2 L^{1/d}) -4 \le -7/2$ if $d \le 2$.

Summarizing,
the Hamiltonian
\begin{align}
\label{eq:finalAQCH}
H' = \frac 1 {L^{1/d}} \left[ \sum_{k=1}^L \Pi_k (\sigma^+_k \sigma^-_0 + \sigma^-_k \sigma^+_0) +
\frac{\sqrt \Delta} 2 \frac{(1+ \sigma_z^0 )} 2 \right] +2 Z 
\end{align}
has $\ket{\psi_0} \ket{10\ldots0}$ as unique eigenstate of eigenvalue 0,
and all other eigenvalues are at distance at least $\sqrt{\Delta}/L^{1/d}$
if $d \le 2$. $H'$ can be regarded as a physical Hamiltonian.
In Fig.~\ref{fig:localH} we represent $H'$ for the case in which $H$
is the frustration-free Hamiltonian of a system in one spatial dimension, and $d=1$.

 \begin{figure}[ht]
  \centering
  \includegraphics[width=7cm]{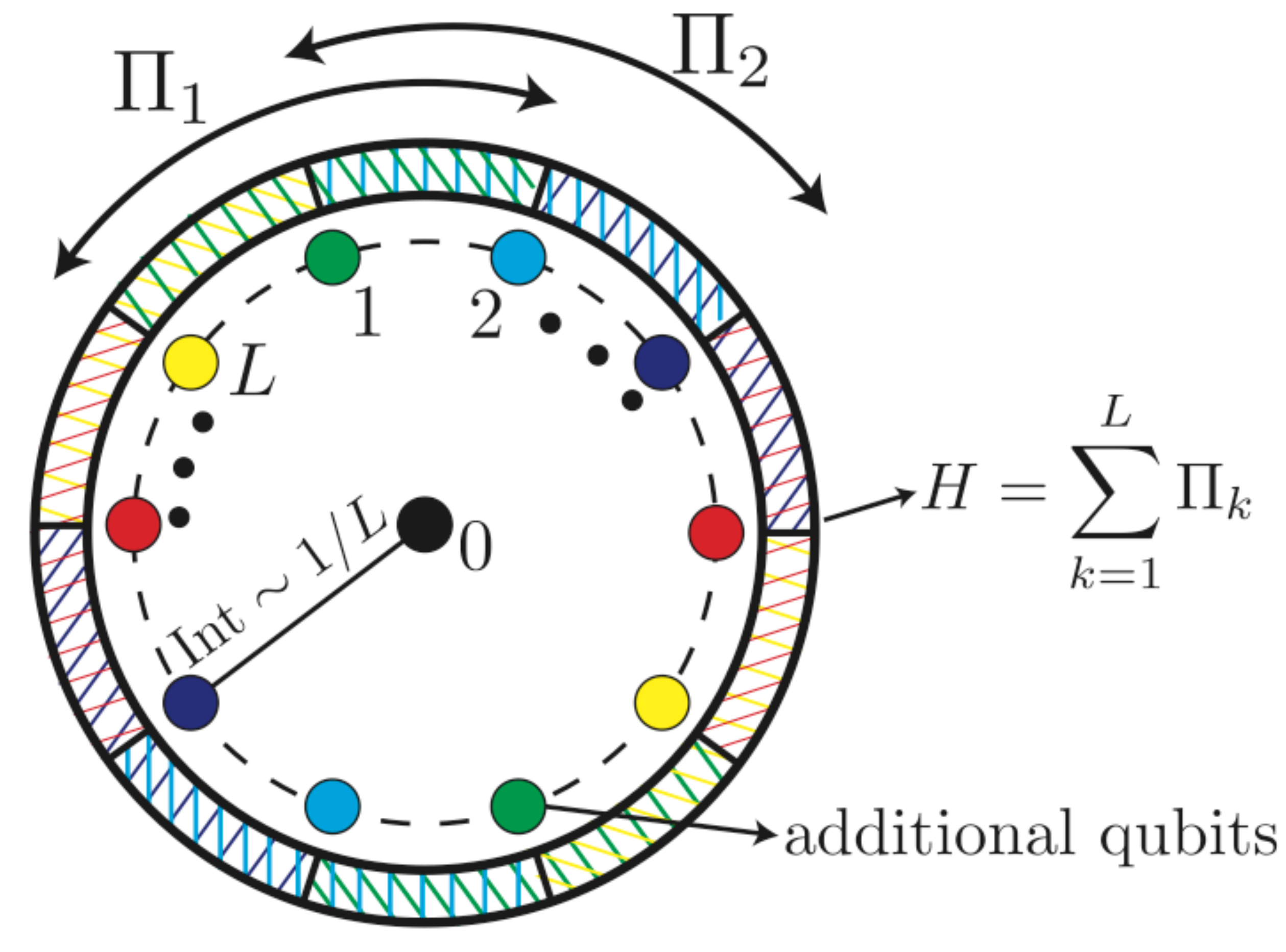}
 \caption{
A physically local representation of $H'$. Each additional qubit $1,\ldots,L$ 
has assigned a color and interacts with a particular term $\Pi_k$ in the Hamiltonian
and the centered qubit. In a one dimensional construction, $\Pi_k$ represents a nearest-neighbor
interaction. Therefore, $\Pi_k$ and $\Pi_{k+1}$ have support overlapping on a single body that is 
represented by lines of color that match that of the additional qubit.}
  \label{fig:localH}
\end{figure}


The construction in Ref.~\cite{aharonov_line_2009}
leads indeed a family of frustration-free Hamiltonians in one spatial dimension.
For a quantum circuit with $Q$ gates,
the spectral gap of $H$ in Ref.~\cite{aharonov_line_2009} can be made of order $1/Q^2$
\footnote{We acknowledge discussions with S. Irani.}.
  $H'$ in Eq.~\eqref{eq:finalAQCH} allows us to amplify the gap 
up to order $1/Q^{1+1/d}$, with an important gain if $d = 2$ and $Q \gg1 $.

\section{GAP: General Hamiltonians}
\label{generalcase}
We show that spectral gap amplification is not possible for general Hamiltonians
in the black-box model.
Again, we consider instances $H_x$ whose unique ground states $\ket{\psi_x}$
can be used to solve SEARCH as  in Lemma~\ref{lemma:random_grover}.
As before, the cost of an evolution with $H_x$ is the number of
 black boxes $O_{H_x}$ required to implement it. We obtain:
\begin{theorem}
  There are stoquastic (sparse) Hamiltonians for which spectral gap amplification
  is not possible. 
\end{theorem}
\begin{proof}
The result will follow again from the known lower bound on SEARCH.
  We consider the family
  $\{H_x \}_{0 \le x \le M-1}$ with
\begin{align}
\label{squarelattice}
H_x = - \ket x \bra x - c A \; ,
\end{align}
and $x$ unknown. 
The constant $c>0$ will be specified later and $A$ is a $M \times M$ symmetric matrix
proportional to the adjacency matrix of a
square lattice in five spatial dimensions. The matrix entries of $A$
are $a_{y'y}=1$ if $(y,y')$ is an edge in the lattice
and $a_{y'y}=0$ otherwise.
$A$ is sparse and independent
of the unknown $x$. In this case we assume a matrix oracle $W_x$ for $H_x$.
This is stronger than $O_{H_x}$ in the sense that $O_{H_x}$ can be implemented
with a single call to $W_x$. Further, by construction, $W_x$ requires calling the search
oracle of Eq.~(\ref{searchoracle}) only once.

We present some useful properties of $H_x$ and refer to
Ref.~\cite{childs_spatialsearch_2004} for details.
There is a {\em phase transition point} at $c=c^*>0$ where the spectral gap
is
\begin{align}
\Delta \to \frac {I_1}{2\sqrt{I_2 M}} \in \cO(1 /\sqrt{M}) \; ,
\end{align}
and
 \begin{align}
\label{probx2}
p_x &=| \langle x \ket{\psi_x}|^2  \to   0.37\;,\\
\label{probs}
p_s & = | \langle s \ket{\psi_x}|^2  \to \frac 1 2\;.
\end{align}
Then, the ground state $\ket{\psi_x}$ has good overlap with both
$\ket x$ and $\ket s$ as in Lemma~\ref{lemma:random_grover}. 
SEARCH can  be solved with cost $\cO(\sqrt{M})$ using the black boxes $O_{H_x}$
or, similarly, the search oracles.
If gap amplification for $H_x$ were possible, this would yield a faster
method to solve SEARCH.
The proof follows by contradiction.
\end{proof}

\section{Stoquastic vs. stochastic matrices}
\label{stochastic-stoquastic}
Because stoquastic Hamiltonians do not suffer from the so-called
sign problem~\cite{suzuki_qmc_1998}, an important task is to devise classical
probabilistic methods to solve the AST problem by sampling from the
distribution as determined by the ground state. 
A well-known method for this purpose is quantum Monte-Carlo (QMC)~\cite{anderson_montecarlo}.
Another method may be obtained by finding a diagonal similarity transformation that maps
stoquastic Hamiltonians to stochastic matrices~\cite{bravyi_stoquastic_2009}.
The fixed point of the resulting Markov chain is the desired distribution.
Contrary to QMC, this second method requires knowledge
of the ratio of amplitudes of the ground state for the mapping.
Our result on no gap amplification for all stoquastic Hamiltonians implies
that finding the transformation to stochastic matrices should be hard in the 
black-box model. Otherwise it would imply a method to amplify the gap
by following the results on frustration-free Hamiltonians.

We can go further and give a family of stoquastic Hamiltonians
for which  classical probabilistic methods, that solve the AST 
by sampling from the ground state, has a large black-box implementation cost.
For some QMC algorithms, our result implies that the spectral gap
of the constructed stochastic matrix rapidly decreases with the problem size.
Once more, our proof considers a family of Hamiltonians $\{ H_x \}_{0 \le x \le M-1}$ 
such that sampling from the ground state of $H_x$ outputs the unknown $x$ with large probability.

\begin{lemma}
  Consider the family of Hamiltonians
\begin{align}
H_x = - \ket x \bra x + F/4 \; ,
\end{align}
where $F \in \mathbb{C}^M \times \mathbb{C}^M$ is independent of $x$ and $\| F \| \le 1$. 
 Then, the lowest eigenvalue is bounded as
$E_0 \le -3/4$, the spectral gap satisfies
$\Delta \ge 1/2$, and the ground state 
has large support  on
$\ket x$, i.e. $|\bra x \psi_x \rangle|^2 \in \cO(1)$.
\end{lemma}

\begin{proof}(Sketched)
 All results can be simply proved by considering $-\ket x \bra x + g F/4$
  and using perturbation theory in the range $[0,1]$ for $g$. 
  \end{proof}

If $F$ is stoquastic and irreducible, the ground state of $H_x$ is $\ket{\psi_x} = \sum_y \sqrt{\pi_y} \ket y$,
with $\pi_y > 0$.
In this case we let $S_x$ be a stochastic matrix whose fixed point is almost the distribution $\pi$.
We assume that a matrix oracle for $S_x$ follows from a matrix oracle for $H_x$.
The latter requires using the (classical) search oracle only once: for input $y$, the matrix elements
of $H_x$ in the $y$th row depend on whether $x=y$ or not.
In path integral QMC, for example,   $S_x$
can be constructed efficiently as follows. For
sufficiently large $\beta$ the probability $\pi_y$ is (almost)
proportional to
\begin{align}
\bra y e^{-\beta H_x} \ket y \; .
\end{align}
Because $\Delta \ge 1/2$, it suffices to choose $\beta \in \cO(\log M)$
to bound the approximation error. Next, because the exponentiation of $H_x$ might be hard to compute,
 we write
$e^{- \beta H_x} = \prod_{l=1}^p e^{-\eta H_x}$, with $\eta \in
\cO(1/\beta)$,
and approximate $e^{-\eta H_x} = \one - \eta H_x + \cO(\eta^2)$. 
We can relate $\prod_{l=1}^p  (\one - \eta H_x)$ with a ``classical system'' $H_c$
(diagonal Hamiltonian) that acts on a larger Hilbert space of dimension
$ M^ L$. More precisely, we define
\begin{align}
   H_c = - \sum \log (f_{r-1,r}) \ket{z^{(r-1)}} \ket{z^{(r)}} \bra{z^{(r-1)}} \bra {z^{(r)}}\;,
\end{align}
with $f_{r-1,r} = \bra{z^{(r-1)}} \one - \eta H_x \ket{z^{(r)}}$.
Note that $\tr(e^{- \beta H_x}) \approx \tr (e^{-H_c})$.
The computation of any diagonal entry of $H_c$ requires
evaluating $p \in \cO[\text{polylog}(M)]$ matrix elements of $H_x$
and thus the same number of search oracles.
Sampling from the distribution determined by $H_c$ requires choosing
an appropriate Markov chain, like the one in the Metropolis
algorithm. From known complexity bounds on the search problem,
any choice of Markov chain will lead, in the black-box model, to a mixing time
\begin{align}
\tau_{\text{mix}} \in \cO\left [ \frac M {\text{polylog}(M)} \right]\; .
\end{align}

\section{Conclusions}
\label{conc}
We introduced and studied the spectral gap amplification problem or GAP.
The goal is to efficiently construct a Hamiltonian
that has the ground state of a given Hamiltonian as eigenstate,
but  a bigger spectral gap. This problem
is motivated by the results in Ref.~\cite{somma_quantum_2008},
where we showed
quantum speedups of Monte-Carlo methods by
giving a quadratic
gap amplification for a specific family of Hamiltonians.
The GAP
has important applications in adiabatic quantum computing and other methods
for adiabatic state transformations, where the
implementation cost is dominated by the inverse gap: If gap amplification is possible, this
results in a quantum speedup. 

We showed that a quadratic spectral gap amplification is indeed
possible for frustration-free Hamiltonians. This
generalizes the results in Ref.~\cite{somma_quantum_2008} and gives 
a more efficient construction to speedup classical Markov-chain based
methods. In a suitable black-box model,
we proved that our method provides the biggest amplification possible.
The result
for frustration-free Hamiltonians also provides more efficient quantum adiabatic methods
to simulate quantum circuits and for the preparation of projected entangled pair states
or PEPS.

We presented a family of
Hamiltonians for which spectral gap amplification
is not possible under some assumptions. The Hamiltonians in this case can be
sparse and stoquastic. They are constructed so that by
preparing their ground states one can solve the unstructured
search problem, and then use  known lower complexity bounds
 to prove the impossibility result by contradiction.
 A corollary is that classical probabilistic methods
that sample from the ground state distribution of
some Hamiltonians may require a large implementation
cost.

\section*{Acknowledgments}
RS thanks Aaron Allen (UNM), Jonas Anderson (UNM), Stephen Jordan
(NIST), and Barbara Terhal (IBM) for discussions on constructions of
local Hamiltonians, A. Childs (IQC) for pointing out
Ref.~\cite{childs_spatialsearch_2004}, and Andrew Landahl for
discussions on constructions of local Hamiltonians.  RS acknowledges
support from the National Science Foundation through the CCF program,
and the Laboratory Directed Research and Development Program at Los
Alamos National Laboratory and Sandia National Laboratories.  SB
thanks the National Science Foundation for support under grant
PHY-0803371 through the Institute for Quantum Information at the
California Institute of Technology and Lockheed Martin.


\appendix

\section{Proof of optimal amplification}
\label{App:OptimalityProof}

  Each $H_x$ is a symmetric matrix related to the adjacency matrix $A$ of
  a sparse expander graph $G=(M,d,\lambda)$.  $M$ is the
  number of vertexes, $d \in \cO(1)$ is the degree of $G$, and
  $\lambda d$ is the second largest eigenvalue of $A$. WLOG, $\lambda \le 1/2$
  (e.g., a Ramanujan graph). 
 Before specifying $H_x$ in a frustration-free form, 
 we define $H_x$ by its matrix elements as follows:
  \begin{align}
\label{matrixentries}
 \bra y H_x \ket z= \left\{ \begin{array}{cl}
 \frac{1}{M-1}   & \text{if} \ y=z=x \cr
 \frac {-1} {d \sqrt{M-1}} & \text{if} \ \{y,z\} \in E \ \textrm{and either $y=x$ or $z=x$}\cr
 \frac {-1} d &  \text{if} \ \{y,z\} \in E \ \textrm{and}\ x\ne y \ne z \ne x \cr
1 & \text{if} \   y = z \ne x \cr
0 & \text{otherwise} .
\end{array}
\right.
\end{align}
$E$ is the set of edges of the graph $G$.
$H_x$ has
the following unique ground state of eigenvalue 0:
\begin{align}
\ket{\psi_x} = \frac 1 {\sqrt{2}} \ket x + \frac 1{ \sqrt{2(M-1)} } \sum_{y \ne x} \ket y \; .
\end{align}
This state satisfies the requirements of the algorithm in
Lemma~\ref{lemma:random_grover}: $\ket{\psi_x}$  has constant overlap with
 $\ket x$ and the equal superposition state.  Since $x$ is unknown,
 the preparation
 of $\ket{\psi_x}$ will allow us to solve SEARCH.

 We can bound the spectral gap $\Delta$ of $H_x$ in terms of $M$.
We define $Z_x = \one - H_x -
\ket {\psi_x} \bra {\psi_x}$. Because $\|H_x \| \le 1$, the norm of $Z_x$ is
$\| Z_x \| = \sup_{\ket \phi} \bra \phi Z_x \ket \phi = 1- \Delta $.
We can write $ \ket \phi = \beta \ket {x} +
\sqrt{1-\beta^2} \ket{x^\perp} $, with $\beta>0$ and
 $\ket{x^\perp}$ a unit state  orthogonal to $\ket x$, to
 obtain
\begin{align}
\nonumber
&| \bra \phi Z_x \ket \phi |\le  2 \beta  \sqrt{1-\beta^2} |\bra x  Z_x \ket {x^\perp} |\;+ \\
\label{gapbound11}
& \ \ \ + \beta^2 | \bra x  Z_x \ket x   |+
(1- \beta^2) | \bra {x^\perp}  Z_x \ket{ x^\perp}| \; .
\end{align}
An upper bound on the {\em rhs}  of Eq.~(\ref{gapbound11}) implies a lower bound on the gap.
We write
$\ket{x^\perp} = \sum_{y \ne x} l_y \ket y$.
The first term on the {\em rhs} of
Eq.~(\ref{gapbound11}) is
\begin{align}
2 \beta \sqrt{1- \beta^2} \sum_{y \ne x} |l_y | \Big|\left [ \bra x
  H_x \ket y +
  \frac 1 {2 \sqrt{M-1}} \right ] \Big| \; .
\end{align}
We use 
\begin{align}
\left| \bra x H_x \ket y+ \frac 1 {2 \sqrt{M-1}} \right| \le \frac 1 {2 \sqrt{M-1}} \ \forall \ y \;
\end{align}
and 
\begin{align}
  \sum_{y \ne x} |l_y| = \|\ket{x^\perp}\|_1 \le \sqrt{M-1}\| \ket{x^\perp}\|_2 = \sqrt{M-1} 
\end{align}
 to show
\begin{align}
\label{firstterm}
|2 \beta  \sqrt{1-\beta^2} \bra x  Z_x \ket {x^\perp} |\le |2  \beta \sqrt{1- \beta^2} | \frac 1 2 \; .
\end{align}
Further, from
the diagonal matrix elements
of $Z_x$, we have ($M \ge 3$):
\begin{align}
\label{secondterm}
 \beta^2 | \bra x  Z_x \ket x   | \le \beta^2 \left[ \frac 1 2 - \frac 1 {M-1} \right] \; .
\end{align}
Finally, if $P_{\bar x} = \sum_{y \ne x } \ket y \bra y$ is the projector onto
the orthogonal subspace of $\ket x$, then the third term in the {\em rhs} 
of Eq.~(\ref{gapbound11})
is bounded by
$(1-\beta^2) \| P_{\bar x} Z_x P_{\bar x} \| $. In addition,
\begin{align}
P_{\bar x} H_x P_{\bar x}  = P_{\bar x}(\one -  A/d) P_{\bar x}  \; ,
\end{align}
where $A$ is the adjacency matrix of  
$G(M,d,\lambda)$. If $\ket s$ is the equal superposition state,
$A \ket s =d \ket s$,
and this is the only eigenstate of $A$ with largest eigenvalue $d$.
We can write
\begin{align}
P_{\bar x} \ket {\psi_x} =\frac{\sqrt{M}} {\sqrt{2(M-1)}} P_{\bar x} \ket s \; ,
\end{align}
and consequently,
\begin{align}
\nonumber
\| P_{\bar x} Z_x P_{\bar x} \| &=\Big\| P_{\bar x} \left(A/d  - \frac M {2(M-1)}
\ket s \bra s \right) P_{\bar x} \Big\| \\
\nonumber
& \le \Big\|   \left(A/d  - \frac M {2(M-1)} \ket s \bra s \right)   \Big\| \; .
\end{align}
Because the second largest eigenvalue
of $A/d$ is $\lambda \le 1/2$, we obtain
\begin{align}
  \label{thirdterm}\| P_{\bar x} Z_x P_{\bar x} \| \le 1-  \frac M
  {2(M-1)}  \le \frac 1 2\;.
\end{align}

Using Eqs.~(\ref{firstterm}),~(\ref{secondterm}),
and~(\ref{thirdterm}), we get
\begin{align}
\label{finalbound1}
  \bra \phi Z_x \ket \phi \le \frac 1 2 + \beta \sqrt{1-\beta^2} - \frac{\beta^2}{M-1} \;.
\end{align}
The maximum  is found when $\beta$ satisfies
\begin{align}
\label{finalbound2}
- \frac{\beta^2}{M-1} = \frac {\beta \sqrt{1- \beta^2}}{(M-1)^2} - \frac 1 {2 (M-1)} \; .
\end{align}
Inserting Eq.~(\ref{finalbound2}) in Eq.~(\ref{finalbound1}),
and since $\beta \sqrt{1-\beta^2} \le 1/2$, yields
\begin{align}
  \bra \phi Z_x \ket \phi 
  & \le 1 - \frac 1 {2(M-1)} +\frac 1 {2(M-1)^2} \\ &\le 1 - \frac 1 {4(M-1)}\;.
\end{align}
for $M \ge 3$.
This implies
\begin{align}
\label{gapestimate}
 \Delta \ge \frac 1 {4 (M-1)} \; .
\end{align}

We now focus on the specification of $H_x$ (or a modified version $\tilde H_x$)
as a frustration-free Hamiltonian.
To this end, we first describe a particular
choice for the projectors appearing in $H_x$.
We define the function
\begin{align}
  c(y) = \left\{
    \begin{array}{cl}
      \frac 1 {\sqrt{d(M-1)}} & \text{if}\ y = x \cr
      \frac 1 {\sqrt d} & \textrm{e.o.c.}
    \end{array}\right.
\end{align}
We also define $E_o$  as the set of ordered pairs $(y,z)$ of edges in $G$, where each
edge appears only once. That is, if $\{y,z\} \in E$, then either $(y,z)$
or $(z,y)$ appears in $E_o$. 
For each edge $(y,z) \in E_0$ we
define the (unnormalized) state
\begin{align}
\ket{\phi_{(y,z)}} = c(y)\ket y - c(z) \ket z\;.
\end{align}
Then,
\begin{align}
\label{Hx}
H_x = \sum_{(y,z) \in E_o} \ket{\phi_{(y,z)}} \bra{\phi_{(y,z)}} \; ,
\end{align}
as the matrix entries coincide with those in Eq.~(\ref{matrixentries}).
Because $\bra{\phi_{(y,z)}} \psi_x \rangle=0$  for all edges, $H_x$ is frustration free.

For consistency with the definitions in previous sections, we define the  projectors
\begin{align}
\Pi_{(y,z)}= \frac {\ket{\phi_{(y,z)}}\bra{\phi_{(y,z)}}} {\bra{\phi_{(y,z)}} \phi_{(y,z)}\rangle}
\end{align}
and the corresponding modified Hamiltonian
\begin{align}
\label{Hx2}
\tilde H_x = \sum_{(y,z) \in E_o} \Pi_{(y,z)} \; .
\end{align}
 $\ket{\psi_x}$ is still the unique
 ground state of $\tilde H_x$ . Since
\begin{align}\label{eq:gap_normalization}
  \bra \varphi \tilde H_x \ket \varphi  = \sum_{(y,z) \in E_o} \bra \varphi  \Pi_{(y,z)} \ket \varphi > \sum _{(y,z) \in E_o} \braket{\varphi}{\phi_{(y,z)}} \braket{\phi_{(y,z)}}{\varphi} =  \bra \varphi H_x \ket \varphi\; ,
\end{align}
the spectral gap of $\tilde H_x $ is also bounded from below as  $\Delta \ge \frac 1 {4 (M-1)}$.

The number of terms in Eq.~(\ref{Hx2}) grows with the number
of edges in the graph,   $d M /2$. It is convenient to reduce the number
of projectors so that $L \in \cO(1)$, and $L$ playing no significant role in the scaling of the amplified gap.
 We then assume a given
 edge coloring of the graph $G$ with edge chromatic number $\chi$
of order $d$~\cite{vizing1964} [i.e., $\chi \in \cO(1)$].  Let
$c_1,\ldots,c_\chi$ be the different colors. All the projectors
$\Pi_{(y,z)}$ belonging to one of the colors are, by construction,
orthogonal to each other as they don't share a vertex. For each $k \in \{1, \ldots, \chi \}$, we define
the projectors
\begin{align}
\Pi_k = \sum_{
  \begin{subarray}{c}
    (y,z) \in c_k 
  \end{subarray}
} \Pi_{(y,z)} \; ,
\end{align}
Then,
$\tilde H_x = \sum_{k=1}^{\chi} \Pi_k $.

The specification of $\tilde H_x$ in this proof is as follows.
For each $k$ we assume the
existence of a reversible version of the matrix oracle $W$ that, on input
$(k,y)$, it outputs $(z,\bra z \Pi_k \ket y)$, where $z$ 
is the only vertex that shares an edge of coloring $c_k$ with $y$.
Note that all the matrix elements $\bra z \Pi_k \ket y$ are equal
unless $y=x$ or $z=x$.  The implementation of $W$
requires then deciding whether its input $y$ is the marked vertex $x$
or not. It implies that the reversible version of $W$ can be 
implemented using a single call to the search oracle of Eq.~(\ref{searchoracle}).
Further, the black box $O_{\tilde H_x}$, which  implements $\exp\{-i \Pi_k t\}$ 
on ancilliary input $\ket k \ket t$, can be constructed
using $\cO(1)$ oracles $W$. This follows from 
the simple observation that $\exp\{-i \Pi_k t\}$ induces rotations
in the two-dimensional subspaces spanned by $\{ \ket y, \ket z \}$,
where $z$ is the only vertex that shares an edge of $c_k$ with $y$.
The implementation of a reversible $O_{\tilde H_x}$ requires then using the search
oracle at most two times ($y \stackrel{?}{=} x$ or $z \stackrel{?}{=}x$).

Let $\tilde H'_x$ be the Hamiltonian that has
$\ket{\psi_x} \ket{\mathfrak 0}$  as a
 non-degenerate eigenstate and spectral gap $\Delta'$. 
 SEARCH can then be solved using the technique in
Sec.~\ref{conventions} which requires evolving with $\tilde H'_x$ 
for time $\cO(1/\Delta')$. By assumption, it implies 
that SEARCH can be solved using 
$\cO(1/\Delta')$ black boxes $O_{\tilde H_x}$ or, similarly,
$\cO(1/\Delta')$ search oracles. Therefore, the 
known complexity bound on SEARCH implies that
$\Delta' \le 1/\sqrt{M} \in \cO(\sqrt{\Delta})$. 
 In this black-box model, our construction
 in Sec.~\ref{frustrationfree} gives the biggest gap amplification
 possible for frustration-free Hamiltonians.


\section{Spectral properties of $\tilde G$}
\label{App:LocalityProof}

The first property to notice is that if $\ket{\psi_0}$ is in the
null space of $H$, then $\Pi_k \ket{\psi_0}=0$.
Therefore $\tilde G \ket{\psi_0}\ket{{0}} =0$ and $\ket{\psi_0}\ket{{0}}$ is an eigenstate
in the null space of $\tilde G$.

As in Sec.~\ref{frustrationfree}, we label the eigenvalues of $H$ by
$\lambda_j$. Assume now $\lambda_j \ne 0$, and consider the action of
$\tilde G$ on the state $\ket{\psi_j 0}$
\begin{align}
  \tilde G \ket{\psi_j 0} =  \sum_k \Pi_k \ket{\psi_j k}\;.  
\end{align}
Notice that
\begin{align}
  \bra{\psi_j 0} \tilde G \ket{\psi_j 0} =   \sum_k \bra{\psi_j k}  \Pi_k \ket{\psi_j 0} = 0
\end{align}
and that
\begin{align}
  \|\tilde G \ket{\psi_j 0} \|^2 = \sum_k \bra{\psi_j k} \Pi_k
    \ket{\psi_j k} =\lambda_j \;.  
\end{align}
We denote by
\begin{align}
  \ket{\perp_j} = \frac 1 {\sqrt{\lambda_j}} \tilde G \ket{\psi_j 0}   
\end{align}
the corresponding normalized state. Next note that
\begin{align}
  \tilde G \ket{\perp_j} = \frac 1 {\sqrt{\lambda_j}} \sum_k \Pi_k
  \ket{\psi_j 0} = \sqrt{\lambda_j} \ket{\psi_j 0}\;.
\end{align}
Therefore, the Hamiltonian $\tilde G$ is invariant in the subspace
$\mathcal{V}_j = \{\ket{\psi_j 0} ,\ket{\perp_j}\}$. Define $\tilde G_j =
\tilde G_{|\mathcal V_j}$ the projection of $\tilde G$ in this invariant
subspace. In the basis $\{\ket{\psi_j 0}, \ket{\perp_j}\}$ we can write
\begin{align}
  \tilde G_j = \begin{pmatrix}
      0 & \sqrt{\lambda_j} \\
      \sqrt{\lambda_j} & 0
    \end{pmatrix}\;,
\end{align}
with eigenvalues $\pm \sqrt{\lambda_j}$.

Finally, we note that because the $\ket{\psi_j}$ form a complete basis, $\tilde G$ acts trivially on any other state orthogonal to $\oplus_j \tilde{\cal V}_j$. While the dimension of the null space of $\tilde G$ is larger than that of the null space of $H$, this has no effects
when solving the AST problem: transitions between $\ket{\psi_j}
\ket{{0}}$ and any other orthogonal state in the null space are
forbidden. We can also use the trick in Sec.~\ref{frustrationfree} to reduce
such a degeneracy, if necessary.

\end{document}